\documentclass{llncs}

\usepackage[latin1]{inputenc}

\usepackage{amssymb}

\usepackage[all]{xy}
\SelectTips{eu}{12}

\usepackage{theorem}
\newtheorem{theo}{Theorem}[section]
\newtheorem{prop}[theo]{Proposition}

\newtheorem{lemm}[theo]{Lemma}
\theorembodyfont{\normalfont}

\newtheorem{defi}[theo]{Definition}

\newcommand{\lb}{\!:\!}
\newcommand{\var}{\bullet}

\newcommand{\bN}{\mathbb{N}} 

\newcommand{\Set}{\mathbf{Set}}

\newcommand{\Gr}{\mathbf{Gr}}

\newcommand{\catC}{\mathbf{C}}

\newcommand{\cN}{\mathcal{N}} % set of nodes 
\newcommand{\cD}{\mathcal{D}} % set of labeled nodes (domain)
\newcommand{\cL}{\mathcal{L}} % labeling function  % changer : fonte ``function'' ? 
\newcommand{\cS}{\mathcal{S}}  % successor function % changer : fonte ``function'' ? 
\newcommand{\ari}{\mathrm{ar}}    % arity function
\newcommand{\setH}{\mathcal{H}} % set of nodes of H 

\newcommand{\parto}{\rightharpoonup}  % partial function 

\newcommand{\catH}{\mathbf{H}}
\newcommand{\HTm}{\catH_{T,m}}  % category of heterogeneous cones 
\newcommand{\CTm}{\catC_{T,m}}  % category of cloning cones 
\newcommand{\tin}{\widetilde{n}} 
\newcommand{\dom}{\mathrm{Dom}} 

\newcommand{\sqgr}{
 \xymatrix@C=6pc{
  L \ar@{-->}[r]_{\tau} \ar[d]_{m} & 
      R \ar[d]^{d} \ar@{--_{>}}@/_3ex/@<-1ex>[l]_{\sigma} \\
  G \ar@{-->}[r]_{\tau_1} & 
      H  \\ }
 }
\newcommand{\sqset}{
 \xymatrix@C=6pc{
  |L| \ar[r]_{\tau} \ar[d]_{|m|} & 
      |R| \ar[d]^{|d|} \ar@{-_{>}}@/_3ex/@<-1ex>[l]_{\sigma} \\
  |G| \ar[r]_{\tau_1} & 
      |H| \ar@{}[ul]|(.2){PO}  \\ }
 }

\title{A Cloning Pushout Approach to Term-Graph Transformation} %% titre !?

\author{D. Duval\inst{1} \and R. Echahed\inst{2} \and F. Prost\inst{2}}
 
\institute{ Laboratoire LJK\\
            B. P. 53,  F-38041 Grenoble, France \\
           \email{Dominique.Duval@imag.fr} \and
            Laboratoire LIG\\
            46, av F\'elix Viallet,
            F-38031 Grenoble, France\\
            \email{Rachid.Echahed@imag.fr}/
            \email{Frederic.Prost@imag.fr}}

\date{October 7., 2008}

\newcommand{\oubli}[1]{} % à supprimer

\begin{document}

%%%%%%%%%%%%%%%%%%%%%%%%%%%%%%%%%%%%%%%%%%%%%%%
%%%%%%%%%%%%%%%%%%%%%%%%%%%%%%%%%%%%%%%%%%%%%%%
\maketitle

%%%%%%%%%%%%%%%%%%%%%%%%%%%%%%%%%%%%%%%%%%%%%%%
% abstract 
\begin{abstract} 
We address the problem of cyclic termgraph rewriting. We propose a new
framework where rewrite rules are tuples of the form $(L,R,\tau,
\sigma)$ such that $L$ and $R$ are termgraphs representing the
left-hand and the right-hand sides of the rule, $\tau$ is a mapping
from the nodes of $L$ to those of $R$ and $\sigma$ is a partial
function from nodes of $R$ to nodes of $L$. $\tau$ describes how
incident edges of the nodes in $L$ are connected in $R$.  $\tau$ is
not required to be a graph morphism as in classical algebraic
approaches of graph transformation.  The role of $\sigma$ is to
indicate the parts of $L$ to be cloned (copied). Furthermore, we introduce a
new notion of \emph{cloning pushout} and define rewrite steps as
cloning pushouts in a given category.  Among the features of the
proposed rewrite systems, we quote the ability to perform local and
global redirection of pointers, addition and deletion of nodes as well
as cloning and collapsing substructures.
\end{abstract}
%%%%%%%%%%%%%%%%%%%%%%%%%%%%%%%%%%%%%%%%%%%%%%%
%%%%%%%%%%%%%%%%%%%%%%%%%%%%%%%%%%%%%%%%%%%%%%%
\section{Introduction}

Complex data-structures built by means of records
and pointers, can formally be represented by \emph{termgraphs}
\cite{BVG87,SPV93,Plu98a}. Roughly speaking, a termgraph is a first-order term with
possible sharing and cycles. The unravelling of a termgraph is a
rational term. Termgraph rewrite systems constitute a high-level
framework which allows one to describe, at a very abstract level,
algorithms over data-structures with pointers. Thus avoiding, on the
one hand, the cumbersome encodings which are needed to translate
graphs (data-structures) into trees in the case of programing with first-order term rewrite systems
and, on the other hand, the many classical errors which may occur in
imperative languages when programing with pointers.

Transforming a termgraph is not an easy task in general.  Many
different approaches have been proposed in the literature which tackle
the problem of termgraph transformation. The algorithmic approach such
as \cite{BVG87} defines in details every step involved in the
transformation of a term-graph by providing the corresponding
algorithm. This approach is too close to implementation
techniques.  In \cite{ArK96}, equational definition
of term-graphs are exploited to define termgraph transformation. These
transformations are obtained up to bisimilar structures (two
termgraphs are bisimilar if they represent the same rational
term). Unfortunately, bisimilarity is not a congruence in general
(e.g., the lengths of two bisimilar but different circular lists are
not bisimilar).
% therefore the equational representation of termgraphs
% cannot be a candidate to achieve our purposes.

A more abstract approach to graph transformation is the
algebraic one, first proposed in the seminal paper
\cite{EhrigPS73}. It defines a rewrite step using the notion of
pushouts. The algebraic approach is quite declarative.  The details of
graph transformations are hidden thanks to pushout constructs. There
are mainly two different algebraic approaches, namely the double
pushout (DPO) and the single pushout (SPO) approaches.

In the DPO approach \cite{EhrigPS73,CorradiniMREHL97}, a rule is defined as a pair of graph morphisms
$L\leftarrow K \rightarrow R$ where $L$, $K$ and $R$ are graphs and
the arrows represent graph homomorphisms.  A graph G rewrites into a
graph H, iff there exists a homomorphism (a matching) $m : L \rightarrow G$ and
a graph D such that the left and the right
squares of the diagram of Fig.\ref{doublepushout} are pushouts.
\begin{figure*}[ht]
\begin{minipage}[b]{0.5\linewidth}
\centering
$$ \xymatrix@C=4pc{
  L \ar[d]_{m} & K \ar[l]_{l} \ar[d]|{d}  \ar[r]^{r} & R \ar[d]^{m'} \\ 
  G & D \ar[l]_{l'} \ar[r]^{r'} & H \\ 
}$$
  \caption{Double pushout: a rewrite step}
  \label{doublepushout}
\end{minipage}
\hspace{0.5cm}
\begin{minipage}[b]{0.5\linewidth}
\centering
$$ \xymatrix@C=4pc{
  L \ar[d]_{m} \ar[r]_{l}& R  \ar[d]^{m'}  \\ 
  G  \ar[r]_{l'} & H  \\ 
}$$
  \caption{Single pushout: a rewrite step}
  \label{spo}
\end{minipage}
\end{figure*}

In general, D is not unique.  Sufficient conditions may be given such
as dangling and identification conditions in order to ensure existence
of pushout complement. The DPO approach is easy to grasp since
morphisms are supposed to be completely defined. However, this
approach fails, in general, to specify rules with deletion of nodes.
For example, if we consider the rule $f(x) \to f(b)$ which can be
translated into the span $f(x) \leftarrow K_0 \rightarrow f(b)$ for some graph $K_0$, and
apply that rule on $f(a)$, then because of pushout properties $f(a)$
is rewritten into a termgraph $H$ which contains $a$. However, $f(b)$ is the only desired result for H.

In the SPO approach \cite{Rao84,Kennaway87,Lowe93,EhrigHKLRWC97}, a
rule is a \emph{partial} graph morphism $L \to R$. When a (total)
graph morphism $m : L \rightarrow G$ exists, $G$ can rewrite to a graph
$H$ iff the  square of Fig~\ref{spo} is a pushout.
%
%\begin{figure*}[ht]
%$$ \xymatrix@C=4pc{
%  L \ar[d]_{m} \ar[r]_{l}& R  \ar[d]^{m'}  \\ 
%  G  \ar[r]_{l'} & H  \\ 
%}$$
%  \caption{Single pushout: a rewrite step}
%  \label{spo}
%\end{figure*}
%
This approach is appropriate to specify deletion of nodes thanks to partial homomorphisms. However, in the case of termgraphs, some care should be taken when a node is deleted. Indeed, deletion of a node causes automatically the deletion of its incident edges. This is not sound in the case of termgraphs since each function symbol should have as many successors as its arity.

%%%%%%%%%%%%%%%%%%%%%%%%%%%%%%
% our objective
%%%%%%%%%%%%%%%%%%%%%%%%%%%%%

In this paper, we investigate a new approach to the definition of
rewrite relations over cyclic termgraphs. We are interested in rewrite
relations, $R$, over termgraphs such that $(t,t')$ belongs to $R$, iff $t'$ is
obtained from t by performing a series of actions of the six following
kinds :(i) addition of new nodes, (ii) redirection of particular
edges, (iii) redirection of all incident edges of a particular node
(iv) deletion of nodes (v) cloning of nodes and (vi) collapsing of
nodes.  In order to deal with these features in a single framework, we
propose  a new algebraic approach to define such rewrite
relations. Our approach departs from the SPO and the DPO approaches. A
rewrite rule is defined as a tuple $(L,R,\tau,\sigma)$ such that, $L$ and $R$ are
termgraphs, respectively the left-hand side and the right-hand side of
the rule. $\tau$ is a mapping from the nodes of $L$ into the nodes of $R$ ($\tau$
has not to be a graph morphism). $\tau(n)= n'$ indicates that incident
edges of $n$ are to be redirected towards $n'$.  $\sigma$ is a partial function
from unlabeled nodes of $R$ into nodes of $L$.  Roughly speaking, $\sigma(n) = p$
indicates that node $n$ should be instantiated as $p$ (parameter
passing). We show that whenever a matching $m:L \to G$ exists, then the
termgraph $G$ rewrites into a termgraph $H$.  We define the termgraph $H$ as an initial
object of a given category. The construction of $H$ could be seen as a
generalization of that of pushouts. We call it \emph{cloning pushout}.

%%%%%%%%%%%%%%%%%%%%%%%%%%%%%%%%
% presentation of the different sections
%%%%%%%%%%%%%%%%%%%%%%%%%%%%%%% 
The paper is organized as follows. In the next section we introduce
the basic definitions of graphs and morphisms that we consider in the
paper. In section~\ref{sec:het}, we introduce a first simplified
version of our rewriting approach. This first step prevents from the
cloning of substructures. Then, in section~\ref{sec:clo}, we give the
full definition of rewriting, including cloning possibility, and
illustrate our approach through several examples in section~\ref{sec:examples}.
Concluding remarks are given in section~\ref{sec:conc}.

%***

%%%%%%%%%%%%%%%%%%%%%%%%%%%%%%%%%%%%%%%%%%%%%%%

% ``The'' diagram (on $\Gr$ then on $\Set$): 
%$$ \sqgr \qquad\qquad \sqset $$

%%%%%%%%%%%%%%%%%%%%%%%%%%%%%%%%%%%%%%%%%%%%%%%
%%%%%%%%%%%%%%%%%%%%%%%%%%%%%%%%%%%%%%%%%%%%%%%
\section{Graphs}
\label{sec:graph}

%[DEP1,DEP2,CG]
In this section we give some technical definitions that we use in the paper. 
We assume the reader is familiar with category theory. The missing definitions may be consulted in \cite{maclane}.

Throughout this paper, a signature $\Omega$ is fixed.
Each operation symbol $\omega\in\Omega$ is endowed with an \emph{arity}
$\ari(\omega)\in\bN$.
For each set $X$, the set of strings over $X$ is denoted $X^*$, and for each
function $f: X \to Y$, the function $f^* : X^* \to Y^*$ is defined by 
$f^*(x_1 \ldots x_n) =f(x_1) \ldots f(x_n)$.

\begin{defi}[Graph]
\label{defi:graph-graph}
A \emph{termgraph}, or simply a \emph{graph} $G=(\cN,\cD,\cL,\cS)$
is made of a set of \emph{nodes} $\cN$ 
and a subset of \emph{labeled nodes} $\cD\subseteq\cN$,  
which is the domain for 
a \emph{labeling function} $\cL:\cD\to\Omega$
and a \emph{successor function} $\cS:\cD\to\cN^*$,
such that for each labeled node $n$, the length of the string $\cS(n)$ 
is the arity of the operation $\cL(n)$. 
For each labeled node $n$ the fact that $\omega=\cL(n)$ is written $n\lb\omega$, 
and each unlabeled node $n$ may be written as $n\lb\var$,
so that the symbol $\var$ is a kind of anonymous variable.

A \emph{graph homomorphism}, or simply a \emph{graph morphism} $g:G\to H$,
where $G=(\cN_G,\cD_G,\cL_G,\cS_G)$ and $H=(\cN_H,\cD_H,\cL_H,\cS_H)$
are graphs, is a function $g:\cN_G\to\cN_H$ 
which preserves the labeled nodes 
and the labeling and successor functions.
This means that $g(\cD_G)\subseteq\cD_H$,
and for each labeled node $n$, $\cL_H(g(n))=\cL_G(n)$ and $\cS_H(g(n))=g^*(\cS_G(n))$
(the image of an unlabeled node may be any node).
This yields the category $\Gr$ of graphs.
\end{defi}

We denote by $\Set$ the classical category of sets.

\begin{defi}[Node functor]
\label{defi:graph-node}
The \emph{node functor} $|-|:\Gr\to\Set$
maps each graph $G=(\cN,\cD,\cL,\cS)$ to its set of nodes $|G|=\cN$ 
and each graph morphism $g:G\to H$ to its underlying function $|g|:|G|\to|H|$.
\end{defi}

We may denote $g$ instead of $|g|$ since the node functor is faithful,
%\cite{DuvalEP07} % [DEP], 
which means that a graph morphism is determined by its underlying
function on nodes.  The faithfulness of the node functor implies that
a diagram of graphs is commutative if and only if its image by the
node functor is commutative, as a diagram of sets.  It may be noted
 % bien que ça ne serve pas 
that the node functor preserves pullbacks,
because it has a left adjoint,
%\cite{CorradiniG99} % [CG], 
and that it does not preserve pushouts.
% for example...
% ex : (n->p) 2 fois, au-dessus de (n) : ça donne (n->p) (une seul p)
% Par \Node :  (n,p) 2 fois, au-dessus de (n) : ça donne (n,p1,p2) (deux p)
%In addition, the set of unlabeled nodes of a graph $G$ is denoted $|G|^\var$.

The following definition introduces a new notion of \emph{graphic functions}.
These functions are used to relate graphs involved in a rewrite step, in addition to classical graph homomorphisms.

\begin{defi}[Graphic functions]
\label{defi:graph-graphic}
Let $G$ and $H$ be graphs and $\gamma:|G|\to|H|$ a function.
For each node $n$ of $G$, 
$\gamma$ is \emph{graphic at $n$} if 
either $n$ is unlabeled or both $n$ and $\gamma(n)$ are labeled, 
$\cL_H(\gamma(n))=\cL_G(n)$ and $\cS_H(\gamma(n))=\gamma^*(\cS_G(n))$.
And $\gamma$ is \emph{strictly graphic at $n$} if
either both $n$ and $\gamma(n)$ are unlabeled or both $n$ and $\gamma(n)$ are labeled, 
$\cL_H(\gamma(n))=\cL_G(n)$ and $\cS_H(\gamma(n))=\gamma^*(\cS_G(n))$.
For each set of nodes $\Gamma$ of $G$, 
$\gamma$ is \emph{graphic (resp. strictly graphic) on $\Gamma$}
if $\gamma$ is graphic (resp. strictly graphic) at every node in $\Gamma$.
\end{defi}
It should be noted that the property of being graphic (resp. strictly graphic) on $\Gamma$
involves the successors of the nodes in $\Gamma$, which may be outside $\Gamma$.

\begin{example}

Let us consider the graphs $G1$ and $G2$ given respectively in Fig~\ref{G1} and Fig~\ref{G2}. Let $\Gamma_1 =\{1,3\}$, $\Gamma_2 =\{1,2,3\}$ and $\Gamma_3 =\{1,2,3,4\}$.
Let $\gamma:|G|\to|H|$ be the  function defined by $\gamma = \{1 \mapsto a, 2 \mapsto b,3 \mapsto c,4 \mapsto d\}$. It is easy to check that
$\gamma$ is graphic on $\Gamma_2$, $\gamma$ is  strictly graphic on $\Gamma_1$,
$\gamma$ is not strictly graphic on $\Gamma_2$ and
$\gamma$ is not graphic on $\Gamma_3$. 

\begin{figure*}[ht]
\begin{minipage}[b]{0.5\linewidth}
\centering

 $$\xymatrix@R=1pc@C=1pc
     {    & 1:f \ar[dl]\ar[d]\ar[dr] \\
    2: \bullet & 3:\bullet & 4:nil
}$$

  \caption{G1}
  \label{G1}
\end{minipage}
\hspace{0.5cm}
\begin{minipage}[b]{0.5\linewidth}
\centering

 $$\xymatrix@R=1pc@C=1pc
     {    & a:f \ar[dl]\ar[d]\ar[dr] \\
    b: nil & c:\bullet & d:\bullet
}$$

  \caption{G2}
  \label{G2}
\end{minipage}
\end{figure*}

\end{example}
Clearly, a function $\gamma:|G|\to|H|$ underlies a graph morphism $g:G\to H$
if and only if it is graphic on $|G|$.
The next straightforward result will be useful.

\begin{lemm}
\label{lemm:graph-graphic}
Let $G$, $H$, $H'$ be graphs and let
$\gamma:|G|\to|H|$, $\gamma':|G|\to|H'|$, $\eta:|H|\to|H'|$ be functions
such that $\gamma'=\eta\circ\gamma$.
Let $\Gamma$ be a set of nodes of $G$. 
If $\gamma$ is strictly graphic on $\Gamma$
and $\gamma'$ is graphic on $\Gamma$,
then $\eta$ is graphic on $\gamma(\Gamma)$.
\end{lemm}

%%%%%%%%%%%%%%%%%%%%%%%%%%%%%%%%%%%%%%%%%%%%%%%
%%%%%%%%%%%%%%%%%%%%%%%%%%%%%%%%%%%%%%%%%%%%%%%
\section{Rewriting without cloning} 
\label{sec:het}

Roughly speaking, in the context of graph rewriting, 
a rewrite rule has a left-hand side graph $L$ and a right-hand side graph $R$,
and a rewrite step applied to a graph $G$ with an occurrence of $L$ 
returns a graph $H$ with an occurrence of $R$,
by replacing $L$ by $R$ in $G$.
We deal with termgraphs, 
so that a labeled node $p$ in $G$ outside $L$ 
and with its $i$-th successor $p'$ in $L$ must have 
some $i$-th successor $n'$ in $H$. For this purpose, 
we introduce a ``target'' function $\tau$,
from the nodes of $L$ to the nodes of $R$, 
and we decide that $n'$ must be $\tau(p')$.
The aim of this section is to define this process precisely.
The definitions and results in this section are simplified versions
of those in the next section.

In this section, 
a \emph{rewrite rule} is tuple $(L, R, \tau)$ made of 
two graphs $L$ and $R$  and a (total) function $\tau:|L|\to|R|$.
A \emph{morphism of rewrite rules}
from $T=(L,R,\tau)$ to $T_1=(L_1,R_1,\tau_1)$
is a pair of graph morphisms $(m,d)$ with $m:L\to L_1$, $d:R\to R_1$ 
such that $|d|\circ\tau=\tau_1\circ|m|$.
% ... category...

In this paper, the illustrations take place either in the category $\Set$ of sets 
% or in the category of graphs,  
or in a heterogeneous framework where the points 
stand for graphs, the solid arrows for graph morphisms 
and the dashed arrows for functions on nodes.
So, a rewrite rule $T=(L,R,\tau)$ can be illustrated as follows:
$$ \xymatrix@C=6pc{
  L \ar@{-->}[r]_{\tau} & 
      R \\ }
$$

It can be noted that 
each graph morphism $t:L\to R$ determines a rewrite rule where $\tau=|t|$.
In this case, for each graph morphism $m:L\to G$
the pushout of $t$ and $m$ in the category $\Gr$ is defined 
as the initial object in the category of cones over $t$ and $m$.
Let us generalize this definition to any rewrite rule $T=(L,R,\tau)$
and any graph morphism $m:L\to G$.
A \emph{heterogeneous cone over $T$ and $m$} is made of
a graph $H$, a function $\tau_1:|G|\to|H|$ and a graph morphism $d:R\to H$ 
such that $T_1=(G,H,\tau_1)$ is a rewrite rule, 
$(m,d):T\to T_1$ is a morphism of rewrite rules 
and $\tau_1$ is graphic on $|G|-|m(L)|$.
$$ \xymatrix@C=6pc{
  L \ar@{-->}[r]_{\tau} \ar[d]_{m} & 
      R \ar[d]^{d} \\
  G \ar@{-->}[r]_{\tau_1} & 
      H \\ }
$$
A \emph{morphism of heterogeneous cones over $T$ and $m$},
say $h:(H,\tau_1,d)\to(H',\tau'_1,d')$, 
is a graph morphism $h:H\to H'$ such that
$|h|\circ\tau_1=\tau'_1$ and $h\circ d=d'$.
This yields the category $\HTm$ of heterogeneous cones over $T$ and $m$.
A \emph{heterogeneous pushout} of $T$ and $m$ is defined as 
an initial object in the category $\HTm$. 

When a heterogeneous pushout exists, 
its initiality property implies that it is unique up to an isomorphism of heterogeneous cones.
A \emph{matching} of a graph $L$ 
is a graph morphism $m:L\to G$ such that $|m|$ is injective. 
It is easy to prove the existence of a heterogeneous pushout 
of a rewrite rule $T=(L,R,\tau)$ and a matching $m:L\to G$, as follows.
Let $(\mathcal{P})$ denote the following pushout of $\tau$ and $|m|$ in $\Set$: 
% dessin PO/Set
$$ \xymatrix@C=6pc{
  |L| \ar[r]_{\tau} \ar[d]_{|m|} & 
    |R| \ar[d]^{\delta} \\
  |G| \ar[r]_{\tau_1} & 
    \setH  \\
}$$
Then $\setH = \tau_1(|G|-|m(L)|) + \delta(|R|)$, 
and in addition the restriction of $\tau_1:|G|-|m(L)|\to\tau_1(|G|-|m(L)|)$ is bijective
and the restriction of $\delta:|R|\to \delta(|R|)$ is bijective.
Hence, a graph $H$ with set of nodes $\setH$ is defined simply by imposing that 
$\tau_1$ is strictly graphic on $|G|-|m(L)|$
and that $\delta$ is strictly graphic on $|R|$.
 It follows that $\delta=|d|$ for a graph morphism $d:R\to H$
and that $(H,\tau_1,d)$ forms a heterogeneous cone over $T$ and $m$.
Now, let us consider any heterogeneous cone $(H',\tau'_1,d')$ over $T$ and $m$.
Because of the pushout of sets $(\mathcal{P})$,
there is a unique function $\eta:|H|\to|H'|$ such that 
$\eta\circ\tau_1=\tau'_1$ and $\eta\circ|d|=|d'|$. 
In addition, it follows from lemma~\ref{lemm:graph-graphic}
that $\eta$ is graphic on $\tau_1(\Gamma)$ and also on $d(|R|)$.
So, $\eta$ underlies a graph morphism $h:H\to H'$.
Since the node functor is faithful, it follows that 
$(H,\tau_1,d)$ is a heterogeneous pushout of $T$ and $m$.

Now, given a rewrite rule $T=(L,R,\tau)$ and a matching $m:L\to G$,
the corresponding \emph{rewrite step} builds the graph morphism
$d:R\to H$, obtained from the heterogeneous pushout of $T$ and $m$.
It can be noted that $d$ is a matching of $R$.  

The induced rewrite relation over termgraphs is unfortunately not
satisfactory. Consider for instance the rule $f(x) \to g(x,
x)$. Informally, the application of such a rule on the termgraph
$1:f(2:a)$ can yield either the termgraph $1:g(2:a,2)$ or the
termgraph $1:g(2:a, 3:a)$ according to the way the term $g(x,x)$ is
represented as a termgraph.  However, the application of the
definition of a rewrite step, as given above, suggests to rewrite the
termgraph $1:f(2:a)$ into $1:g(2:\bullet,2)$ by means of the following
rule $(1:f(x:\bullet),1:g(x:\bullet, x),\tau = \{1 \mapsto 1, x
\mapsto x\})$.  The node $2$ is not labeled in the reduced
termgraph. This reflects the fact that the instance of $x$ cannot be
substituted or cloned in the right-hand side. We overcome this
drawback in the next section.

%\begin{example}[collapse]
%  Collapsing nodes is easily implemented in this formalism.
%We refer to example \ref{ex:graphconv} for the definition of the 
%graphical conventions. Consider for instance the following rule: 
%\begin{center}
% \begin{tabular}{|c|c|} \hline
% \multicolumn{2}{|c|}{$\tau :4 \mapsto 3, 6 \mapsto 5$} \\ \hline 
%  $\xymatrix@C=1pc@R=1pc{ 1:col \ar[d] & 4:g \ar[r] & 6:a  \\
%              2:f \ar[r] \ar[ur] & 3:g \ar[r] & 5:a
%     }$ 
%   & $\xymatrix@C=1pc@R=1pc{ 2:f \ar@/^.5pc/[r] \ar@/^-.5pc/[r] & 
%       3:g \ar[r] & 5:a}$  \\ \hline
  
% \end{tabular}
%\end{center}

%\end{example}

%%%%%%%%%%%%%%%%%%%%%%%%%%%%%%%%%%%%%%%%%%%%%%%
%%%%%%%%%%%%%%%%%%%%%%%%%%%%%%%%%%%%%%%%%%%%%%%
\section{Rewriting with cloning} 
\label{sec:clo}

In this section, the definitions and results of the previous
section are generalized in order to add a ``cloning'' process.
Indeed, in the resulting graph $H$ from section~\ref{sec:het} 
there is no node in $R$ with its image outside $R$.
This is an issue, which is solved in this section 
thanks to the notion of ``clone''. 
Roughly speaking, a clone of a labeled node $p$ in $G$
is a node $n$ in $H$ 
with the same label and ``the same'' successors as $p$,
where ``the same'' successors are defined via 
the target function $\tau$ from the previous section. 
The definition of a rewrite rule is generalized so that it yields 
the information about the way the images of the nodes in $L$
must be cloned by images of nodes in $R$.
The main result is theorem~\ref{theo:clo-po}:
under relevant definitions and assumptions, 
for each rewrite rule $T$ and matching $m$
there is a \emph{cloning pushout} of $T$ and $m$,
which can be built explicitly from a pushout of sets. 
Since each node in $L$ may have an arbitrary number of clones 
(maybe no clone at all), and a node in $R$ cannot be a clone of 
more than one node in $L$, 
the relation between the nodes in $L$ and their clones in $R$
takes the form of a partial function from $|R|$ to $|L|$.
In this paper, partial functions are denoted with the symbol ``$\parto$'',
the domain of a partial function $\sigma$ is denoted $\dom(\sigma)$, 
and the composition of partial functions is defined as usual.

\begin{defi}[Clones]
\label{defi:clo-clone}
Let $G$ and $H$ be graphs and $\tau:|G|\to|H|$ a function.
Then $p\in|H|$ is a \emph{$\tau$-clone} of $q\in|G|$ when: 
$p$ is labeled if and only if $q$ is labeled, 
and then $\cL_H(p)=\cL_G(q)$ and $\cS_H(p)=\tau^*(\cS_G(q))$.
\end{defi}
% donc : strict graphic vs clones...

\begin{defi}[Rewrite rule]
\label{defi:clo-rule}
A \emph{rewrite rule} is tuple $(L,R,\tau,\sigma)$ made of 
two graphs $L$ and $R$,
a function $\tau:|L|\to|R|$ 
and a partial function $\sigma:|R|\parto|L|$ 
such that each node $n$ in the domain of $\sigma$ 
is unlabeled or is a $\tau$-clone of  $\sigma(n)$.
A \emph{morphism of rewrite rules},
from $T=(L,R,\tau,\sigma)$ to $T_1=(L_1,R_1,\tau_1,\sigma_1)$
is a pair of graph morphisms $(m,d)$ with $m:L\to L_1$ and $d:R\to R_1$ 
such that $|d|\circ\tau=\tau_1\circ|m|$,
$d(\dom(\sigma)) \subseteq \dom(\sigma_1)$ 
and $|m|\circ\sigma=\sigma_1\circ|d|$ on $\dom(\sigma)$.
% ... category...
\end{defi}

In the previous section, we have dealt with the simple case 
where the domain of $\sigma$ is empty.

In the sequel, a rewrite rule $T=(L,R,\tau,\sigma)$ will be illustrated as follows:
$$ 
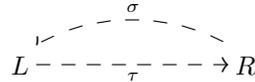
\xymatrix@C=6pc{
  L \ar@{-->}[r]_{\tau} & 
      R \ar@{--_{>}}@/_3ex/@<-1ex>[l]_{\sigma} \\ }
$$

\begin{minipage}{0.90\linewidth}
or depicted as opposite, where the lines $\tau$ and $\sigma$ contain the definitions
        of the functions $\tau$ and $\sigma$. 
\end{minipage}
\hspace{0.2cm}
\begin{minipage}{0.45\linewidth}
%\begin{center}
        \begin{tabular}{|c|c|}
         \hline
         \multicolumn{2}{|l|}{$\tau$:} \\ \hline 
         \multicolumn{2}{|l|}{$\sigma$:} \\ \hline
         L & R \\ \hline
        \end{tabular}
    %    \noindent
%\end{center}
\end{minipage}

%Consider a simple example of rule below. the operation $rc$ updates
%the information $3:\bullet$ accessible from symbol $f$ by the content
%of $4:\bullet$. This is done by operating a local redirection of the
%edge outgoing node $2:f$. This edge points the node $\tau(4) =
%4$ in $R$.
%\begin{center}
%\begin{tabular}{|c|c|}
%  \hline 
%  \multicolumn{2}{|c|}{$\tau : 2 \mapsto 1, i \mapsto i \mbox{ for } i  \in \{2,3,4\}$} \\ \hline 
%  \multicolumn{2}{|c|}{$\sigma : 3 \mapsto 3, 4 \mapsto 4$} \\ \hline
%  $\xymatrix{ 1:rc \ar^{1}[r] \ar^{2}[d]& 2:f \ar[d]\\
%              4:\bullet & 3 : \bullet
%   }$ & 
%  $\xymatrix{ & 2:f \ar[dl] \\ 
%              4:\bullet & 3:\bullet}$ \\ \hline
%\end{tabular}  
%\end{center}

%%  When unspecified, one has $\tau(i)=i$, following our conventions. 
  
%%  In our example rule, the edge from node $1$ to node $4$ 
%%in $L$ is thus the second argument of function $rc$.

\begin{example}[if-then-else]
  \label{ex:ifthenelse}

Below, we give the rewrite rules which define the If-then-else operator as it behaves in classical imperative languages.

\begin{center}
 \begin{tabular}{|c|c|}
  \hline 
  \multicolumn{2}{|c|}{$\tau : 1 \mapsto 5, 
                        2 \mapsto 5, 3 \mapsto 5, 4 \mapsto 5$} \\ \hline 
  \multicolumn{2}{|c|}{$\sigma : 5 \mapsto 3 $} \\ \hline
  $\xymatrix@C=1pc@R=1pc{ & 1:if \ar[d] \ar[dr] \ar[dl] &    \\
             2:true & 3:\bullet & 4:\bullet}$ 
   & $\xymatrix{5:\bullet}$  \\ \hline
 \end{tabular}
~~~~~~~~~~~~
 \begin{tabular}{|c|c|}
  \hline 
  \multicolumn{2}{|c|}{$\tau : 1 \mapsto 5, 
                        2 \mapsto 5, 3 \mapsto 5, 4 \mapsto 5$} \\ \hline 
  \multicolumn{2}{|c|}{$\sigma : 5 \mapsto 4$} \\ \hline
  $\xymatrix@C=1pc@R=1pc{ & 1:if \ar[d] \ar[dr] \ar[dl]     \\
             2:false & 3:\bullet & 4:\bullet}$ 
   & $\xymatrix{5:\bullet}$  \\  \hline
 \end{tabular}
\end{center}
The definition of $\tau$ ensures that the if-then-else expression is replaced by its value $\tau(1) = 5 $. The definition of $\sigma$ indicates that the value of the  if-then-else is its  second (resp. third) argument specified by $\sigma(5)= 3$ (resp. $\sigma(5)= 4$) in the rules above. Notice that if $\sigma$ were defined as the empty function, the if-then-else expression would evaluate to an unlabeled node. 

% given on the left side  and every node of the left handside to 
%the appropriate branch of the if-then-else. On the other hand the definition
%of $\sigma$ indicates that it is the actual value (the one defined through
%the matching) that is going to be the result of the rewrite step. If
%$\sigma$ had been defined as the empty function, then the result of the
%rewrite would only yield an unlabelled node. 
\end{example}

\begin{example}[Cloning data-structures]
  \label{ex:clon}
  In this example we give the rules to clone natural numbers, 
encoded with $succ$ and $zero$. The clone of $zero$ is done
using the following rule:

\begin{center}
 \begin{tabular}{|c|c|}
  \hline 
  \multicolumn{2}{|c|}{$\tau : 1 \mapsto 2, 2 \mapsto 2 $} \\ \hline 
  \multicolumn{2}{|c|}{$\sigma : 3 \mapsto 2, 2 \mapsto 2 $} \\ \hline
  $\xymatrix@C=1pc@R=1pc{ 1:clone \ar[d]    \\
                          2:zero} $ 
   & $\xymatrix{2:zero & 3:zero}$  \\ \hline
 \end{tabular}
\end{center}

One can note that the condition on the labeled nodes (ie $\tau$-clones, 
see def. \ref{defi:clo-rule}) in the 
domain of $\sigma$ is verified. This rule redirects all edges 
from $1$ to $2$, while the edges adjacent to $2$ remain unchanged.  

The second rule is defined as follows:
\begin{center}
 \begin{tabular}{|c|c|}
  \hline 
  \multicolumn{2}{|c|}{$\tau : 1 \mapsto 4, 2 \mapsto 2, 3 \mapsto 3 $} \\ \hline 
  \multicolumn{2}{|c|}{$\sigma : 3 \mapsto 3$} \\ \hline
  $\xymatrix@C=1pc@R=1pc{ 1:clone \ar[r] & 2:succ \ar[d]    \\
                          & 3:\bullet} $ 
   & $\xymatrix@C=1pc@R=1pc{2:succ \ar[d] & 4:succ \ar[d] \\
                3:\bullet & 5:clone \ar[l]}$  \\ \hline
 \end{tabular}
\end{center}

Notice that, in this case, it is not possible to define $\sigma(4)=2$
because sucessor of $4$ in $R$ is labeled by $clone$ and successor of $2$ in 
$L$ is labeled by $succ$, thus breaking the $\tau$-clone condition.
\end{example}

\begin{defi}[Cloning cone]
\label{defi:clo-cone}
Let $T=(L,R,\tau,\sigma)$ be a rewrite rule and $m:L\to G$ a graph morphism. 
A \emph{cloning cone over $T$ and $m$} is  a tuple $(H, \tau_1, d, \sigma_1)$ made of 
a graph $H$, a function $\tau_1:|G|\to|H|$, a graph morphism $d:R\to H$ 
and a partial function $\sigma_1:|H|\parto|G|$ 
such that $T_1=(G,H,\tau_1,\sigma_1)$ is a rewrite rule, 
$(m,d):T\to T_1$ is a morphism of rewrite rules, 
$\tau_1$ is graphic on $|G|-|m(L)|$
and $n_1$ is a $\tau_1$-clone of $\sigma_1(n_1)$ for each $n_1$ in the domain of $\sigma_1$.
$$ 
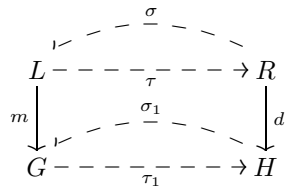
\xymatrix@C=6pc{
  L \ar@{-->}[r]_{\tau} \ar[d]_{m} & 
      R \ar[d]^{d} \ar@{--_{>}}@/_3ex/@<-1ex>[l]_{\sigma} \\
  G \ar@{-->}[r]_{\tau_1} & 
      H \ar@{--_{>}}@/_3ex/@<-1ex>[l]_{\sigma_1} \\ }
$$
A \emph{morphism of cloning cones over $T$ and $m$},
say $h:(H,\tau_1,d,\sigma_1)\to(H',\tau'_1,d',\sigma'_1)$,
is a graph morphism $h:H\to H'$ such that 
$|h|\circ\tau_1=\tau'_1$, $h\circ d=d'$, 
$h(\dom(\sigma_1)) \subseteq \dom(\sigma'_1)$ and 
$\sigma'_1\circ|h|=\sigma_1$ on $\dom(\sigma_1)$.
\\ This yields the category $\CTm$ of cloning cones over $T$ and $m$.
\end{defi}

\begin{defi}[Cloning pushout]
\label{defi:clo-pushout}
Let $T=(L,R,\tau,\sigma)$ be a rewrite rule and $m:L\to G$ a graph morphism. 
A \emph{cloning pushout of $T$ and $m$} 
is an initial object in the category $\CTm$ of cloning cones over $T$ and $m$.
\end{defi}

When a cloning pushout exists, 
its initiality implies that it is unique up to an isomorphism of cloning cones.
In theorem~\ref{theo:clo-po} we prove the existence of a cloning pushout of $T$ and $m$ 
under some injectivity assumption on $m$. 

\begin{defi}[Matching]
\label{defi:clo-match}
A \emph{matching} with respect to a rewrite rule $T=(L,R,\tau,\sigma)$
is a graph morphism $m:L\to G$ such that 
if $m(p)=m(p')$ for distinct nodes $p$ and $p'$ in $L$ 
then $\tau(p)$ and $\tau(p')$ are in $\dom(\sigma)$ and 
$\sigma(\tau(p))=\sigma(\tau(p'))$ in $L$.
\end{defi}

\begin{prop}
\label{prop:clo-match-set}
Let $T=(L,R,\tau,\sigma)$ be a  rewrite rule 
and $m:L\to G$ a matching with respect to $T$. 
Then the pushout of $\tau$ and $|m|$ in $\Set$: 
$$ \xymatrix@C=4pc{
  |L| \ar[r]_{\tau} \ar[d]_{|m|} & 
    |R| \ar[d]^{\delta} \\
  |G| \ar[r]_{\tau_1} & 
    \setH  \\
}$$
satisfies: 
  $$ \setH = \tau_1(\Gamma) + \delta(\Delta) + \delta(\Sigma) $$
where $\Gamma=|G|-|m(L)|$, $\Sigma=\dom(\sigma)$, $\Delta=|R|-\Sigma$ and:
\begin{itemize}
\item the restriction of $\tau_1:\Gamma\to\tau_1(\Gamma)$ is bijective,
\item the restriction of $\delta:\Delta\to \delta(\Delta)$ is bijective,
\item and the restriction of $\delta:\Sigma\to \delta(\Sigma)$ 
is such that if $\delta(n)=\delta(n')$ for distinct nodes $n$ and $n'$ in $\Sigma$
then $\sigma(n)=\sigma(n')$ in $L$.
\end{itemize}
In addition, there is a unique partial function $\sigma_1:\setH\parto|G|$ 
with domain $\delta(\Sigma)$
such that $|m|\circ\sigma=\sigma_1\circ\delta$. 
\end{prop}

\begin{proof}
Clearly $\setH = \tau_1(\Gamma) + \delta(|R|)$
with the restriction of $\tau_1:\Gamma\to\tau_1(\Gamma)$ bijective.
If  $\delta(n)=\delta(n')$ for distinct nodes $n$ and $n'$ in $R$,
then there is a chain from $n$ to $n'$ made of pieces like this one: 
$$ 
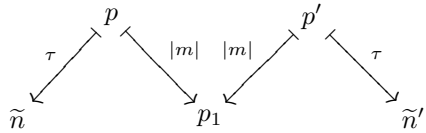
\xymatrix@C=2pc{
  & p \ar@{|->}[dl]_{\tau} \ar@{|->}[dr]^{|m|} && 
    p'  \ar@{|->}[dl]_{|m|} \ar@{|->}[dr]^{\tau} & \\
  \tin && p_1 && \tin' \\
}$$
with $\tin,\tin'\in|R|$, $p,p'\in|L|$, $p_1\in|G|$, 
and it can be assumed that $\tin\ne\tin'$ and $p\ne p'$. 
Since $m$ is a matching, 
$\tin$ and $\tin'$ are in $\Sigma$ and $\sigma(\tin)=\sigma(\tin')$.
The decomposition of $\setH$ follows.
\\ 
Now, let $n_1\in\delta(\Sigma)$ and let us choose some $n\in\Sigma$ 
such that $n_1=\delta(n)$. If $\sigma_1$ exists, 
then $\sigma_1(n_1)=\sigma_1(\delta(n))=m(\sigma(n))$. 
On the other hand, if $n'\in\Sigma$ is another node 
such that $n_1=\delta(n)$, then we have just proved that 
$\sigma(n)=\sigma(n')$, so that $m(\sigma(n))$
does not depend on the choice of $n$, it depends only on $n_1$.
So, there is a unique $\sigma_1:\setH\parto|G|$ as required, 
it is defined by $\sigma_1(n_1)=m(\sigma(n))$ for any $n\in\Sigma$ 
such that $n_1=\delta(n)$.
\end{proof}

\begin{prop}
\label{prop:clo-match-gr}
Let $m:L\to G$ be a matching with respect to a rewrite rule $T=(L,R,\tau,\sigma)$.
The pushout of $\tau$ and $|m|$ in $\Set$,
with $\sigma_1$ as in proposition~\ref{prop:clo-match-set}, 
underlies a cloning cone over $T$ and $m$.
% a unique ? ou bien unicité conséquence du theo ?
\end{prop}

\begin{proof}
First, let us define a graph $H$ with set of nodes $\setH$.
According to proposition~\ref{prop:clo-match-set}, 
and with the same notations, 
a graph $H$ with set of nodes $\setH$ is defined by imposing that 
$\tau_1$ is strictly graphic on $\Gamma$,
that $\delta$ is strictly graphic on $\Delta$, 
and that each node $n_1\in\delta(\Sigma)$ is a $\tau_1$-clone of $q_1$, 
where $q_1=\sigma_1(n_1)$. 

Now, let us prove that $\delta$ underlies a graph morphism $d:R\to H$.
Since $\delta$ is graphic on $\Delta$, 
we have to prove that $\delta$ is also graphic on $\Sigma$.
Let $n\in\Sigma$ and $n_1=\delta(n)$.
If $n$ is unlabeled there is nothing to prove,
otherwise let $q=\sigma(n)$, 
then $q$ is labeled, $\cL_R(n)=\cL_L(q)$ and $\cS_R(n)=\tau^*(\cS_L(q))$.
Then $m(q)=m(\sigma(n))=\sigma_1(\delta(n))=q_1$,
and from the fact that $m$ is a graph morphism we get 
$\cL_L(q)=\cL_G(q_1)$ and $|m|^*(\cS_L(q))=\cS_G(q_1)$.
The definition of $H$ imposes 
$\cL_G(q_1)=\cL_H(n_1)$ and $\tau_1^*(\cS_G(q_1))=\cS_H(n_1)$.
Altogether, $\cL_R(n)=\cL_H(n_1)$ and 
$\cS_H(n_1)=(\tau_1^*(|m|^*(\cS_G(q)))=
\delta^*(\tau^*(\cS_G(q)))=\delta^*(\cS_R(n))$,
so that indeed $\delta$ is also graphic on $\Sigma$.

Finally, it is easy to check that  this yields a cloning cone over $T$ and $m$.
\end{proof}

\begin{theo}
\label{theo:clo-po}
Given a rewrite rule $T=(L,R,\tau,\sigma)$ and a matching $m:L\to G$
with respect to $T$, 
the cloning cone over $T$ and $m$ defined in proposition~\ref{prop:clo-match-gr}
is a pushout of $T$ and $m$. 
\end{theo}

\begin{proof}
The cloning cone over $T$ and $m$ from proposition~\ref{prop:clo-match-gr} is denoted
$(m,d):T\to T_1$ with $T_1=(G,H,\tau_1,\sigma_1)$. 
Let us consider any cloning cone over $T$ and $m$, say  
$(m,d'):T\to T'_1$ with $T'_1=(G',H',\tau'_1,\sigma'_1)$. 
Since $(m,d)$ underlies a pushout of sets, 
there is a unique function $\eta:|H|\to|H'|$ such that $\eta\circ|d|=|d'|$ and 
$\eta\circ\tau_1=\tau'_1$. 
Let $\Sigma=\dom(\sigma)$ and $\Sigma_1=\dom(\sigma_1)$. 
Because the node functor is faithful, the result will follow if we can prove 
that $\eta(\Sigma_1)\subseteq\Sigma'_1$ and $\sigma'_1\circ\eta=\sigma_1$ on $\Sigma_1$, 
and that $\eta$ underlies a graph morphism.
\\ 
We have $\eta(\Sigma_1)=\eta(d(\Sigma))=d'(\Sigma)\subseteq\Sigma'_1$,
and for each $n_1\in\Sigma_1$, let $n\in\Sigma$ such that $n_1=d(n)$,
then on one hand $\sigma'_1(\eta(n_1))=\sigma'_1(\eta(d(n)))=\sigma'_1(d'(n))=m(\sigma(n))$
and on the other hand $\sigma_1(n_1)=\sigma_1(d(n))=m(\sigma(n))$,
hence as required $\sigma'_1(\eta(n_1))=\sigma_1(n_1)$.
\\ 
In order to check that $\eta$ underlies a graph morphism $h:H\to H'$, 
we use the decomposition of $\setH$ from proposition~\ref{prop:clo-match-set}
and the construction of the cloning cone $(m,d)$ in % the proof of...
proposition~\ref{prop:clo-match-gr}.
It follows immediately from lemma~\ref{lemm:graph-graphic}
that $\eta$ is graphic on $\tau_1(\Gamma)$ and also on $d(\Delta)$.
Let us prove that $\eta$ is graphic on $\Sigma_1$. 
Let $n_1\in\Sigma_1$, $q_1=\sigma_1(n_1)$ and $n'_1=\eta(n_1)$.
Then $q_1=\sigma'_1(n'_1)$ because $\sigma'_1\circ\eta=\sigma_1$.
So, $n_1$ is a $\tau_1$-clone of $q_1$ 
and $n'_1$ is a $\tau'_1$-clone of the same node $q_1$.
This means that $\cL_{H'}(n'_1)=\cL_G(q_1)=\cL_H(n_1)$
and that $\cS_{H'}(n'_1)=(\tau'_1)^*(\cS_G(q_1))=
\eta^*(\tau_1^*(\cS_G(q_1)))=\eta^*(n_1)$.
So, $\eta$ is graphic on $\Sigma_1$, and since $d(\Sigma)\subseteq\Sigma_1$, 
it follows that $\eta$ is graphic on $d(\Sigma)$.
Altogether, $\eta$ is graphic on the whole of $|H|$,
which means that $\eta=|h|$ for a graph morphism $h:H\to H'$.
This concludes the proof.
\end{proof}

\begin{defi}[Rewrite step]
\label{defi:clo-step}
Given a rewrite rule $T=(L,R,\tau,\sigma)$ and a matching 
$m:L\to G$ with respect to $T$, 
the corresponding \emph{rewrite step} 
builds the graph morphism $d:R\to H$,
obtained from the cloning pushout of $T$ and $m$.
\end{defi}

\begin{example}
\label{clone}
We go back to the rule $f(x) \to g(x,x)$ discussed at the end of
section~\ref{sec:het}. This rule can be represented in our framework
in different manners according to the way the term $g(x,x)$ is
represented as a termgraph and also to the way the functions $\tau$
and $\sigma$ are defined. We give below two different rules.  Let $G$
be the termgraph $1:f(2:a)$. The first rule (Rule1) rewrites the termgraph $G$
into $1:g(2:a, 2)$, while the second rule (Rule2) rewrites $G$ into $1:g(2:a,
3:a)$.  The node $2$ and $3$ in $1:g(2:a, 3:a)$ are clones of node $2$
in $G$.
 
%\begin{center}
% \begin{tabular}{|c|c|}
%  \hline 
%  \multicolumn{2}{|c|}{$\tau : 1 \mapsto 1, 2 \mapsto 2$} \\ \hline 
%  \multicolumn{2}{|c|}{$\sigma : 2 \mapsto 2, 3 \mapsto 2$} \\ \hline
%  $\xymatrix@C=1pc@R=1pc{ 1:f \ar[d]     \\
%                          2:\bullet} $ 
%   & $\xymatrix{  1:g \ar@/^.5pc/[d] \ar@/_.5pc/[d]  \\
%                2:\bullet }$  \\ \hline
% \end{tabular}
%\end{center}

%\begin{center}
% \begin{tabular}{|c|c|}
%  \hline 
%  \multicolumn{2}{|c|}{$\tau : 1 \mapsto 1, 2 \mapsto 2$} \\ \hline 
%  \multicolumn{2}{|c|}{$\sigma : 2 \mapsto 2, 3 \mapsto 2$} \\ \hline
%  $\xymatrix@C=1pc@R=1pc{ 1:f \ar[d]     \\
%                          2:\bullet} $ 
%   & $\xymatrix{ 1:g \ar[d]\ar[dr] & \\
%                2:\bullet & 3:\bullet }$  \\ \hline
% \end{tabular}
%\end{center}

\begin{figure*}[ht]
\begin{minipage}[b]{0.5\linewidth}
\centering
 \begin{tabular}{|c|c|}
  \hline 
  \multicolumn{2}{|c|}{$\tau : 1 \mapsto 1, 2 \mapsto 2$} \\ \hline 
  \multicolumn{2}{|c|}{$\sigma : 2 \mapsto 2$} \\ \hline
  $\xymatrix@C=1pc@R=1pc{ 1:f \ar[d]     \\
                          2:\bullet} $ 
   & $\xymatrix@C=1pc@R=1pc{  1:g \ar@/^.5pc/[d] \ar@/_.5pc/[d]  \\
                2:\bullet }$  \\ \hline
 \end{tabular}
  \caption{Rule1}
  \label{rule1}
\end{minipage}
\hspace{0.5cm}
\begin{minipage}[b]{0.5\linewidth}
\centering
 \begin{tabular}{|c|c|}
  \hline 
  \multicolumn{2}{|c|}{$\tau : 1 \mapsto 1, 2 \mapsto 2$} \\ \hline 
  \multicolumn{2}{|c|}{$\sigma : 2 \mapsto 2, 3 \mapsto 2$} \\ \hline
  $\xymatrix@C=1pc@R=1pc{ 1:f \ar[d]     \\
                          2:\bullet} $ 
   & $\xymatrix@C=1pc@R=1pc{ 1:g \ar[d]\ar[dr] & \\
                2:\bullet & 3:\bullet }$  \\ \hline
 \end{tabular}

  \caption{Rule2}
  \label{rule2}
\end{minipage}
\end{figure*}
\end{example}

\section{Examples}
\label{sec:examples}

\begin{minipage}{0.90\linewidth}
In this section, we give some illustrating examples.
We represent a rewrite step $G \to H$ performed using a rewrite rule
$(L,R,\tau,\sigma)$ as in the figure opposite.
We assume in the given examples that the matching  morphism $m : L \to G$ is such that
$m(i) = i$.
\end{minipage}
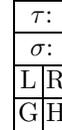
\hspace{0.2cm}
\begin{minipage}{0.45\linewidth}
        \begin{tabular}{|c|c|}
         \hline
         \multicolumn{2}{|c|}{$\tau$:} \\ \hline 
         \multicolumn{2}{|c|}{$\sigma$:} \\ \hline
         L & R \\ \hline
         G & H \\ \hline
 %        \multicolumn{2}{|c|}{$\tau1$:} \\ \hline 
 %        \multicolumn{2}{|c|}{$\sigma1$:} \\ \hline
        \end{tabular}
\end{minipage}

%\begin{example}[Insertion in a circular list]
\section*{Insertion in a circular list}

  In this example we give a rule which defines the insertion of an
  element at the head of a circular list of size greater that one. In
  this rule, node $3$ is the head of the list, and $6$ is the last
  element of the list. Notice that, in $R$, the pointer to the head of
  the list, the second argument of node $6$, has been moved from $3$ to
  the new node $11$ in $R$. The definition of $\tau$ is such that all pointers to
  the head of the list are moved from $3$ to $11$ ($\tau(3)=11$).
 We apply the rule on a circular list of four items.

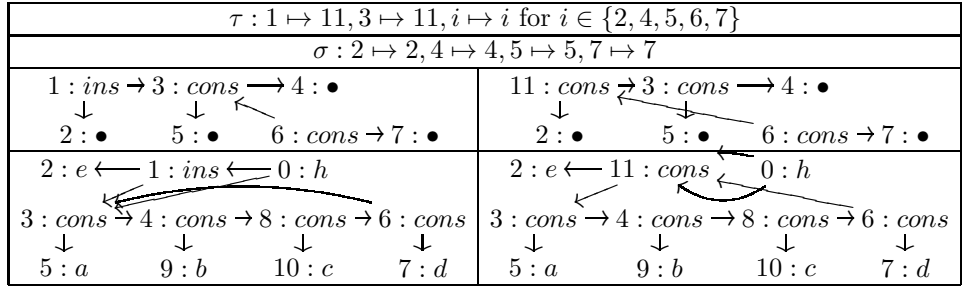
\begin{center}
\begin{tabular}{|c|c|}
\hline
 \multicolumn{2}{|c|}{$\tau:1 \mapsto 11 , 3 \mapsto 11,  i \mapsto i \mbox{ for } i  \in \{2,4,5,6,7\}  $} \\ \hline 
 \multicolumn{2}{|c|}{$\sigma : 2 \mapsto 2, 4 \mapsto 4, 
                                5 \mapsto 5, 7 \mapsto 7$} \\ \hline 
 $\xymatrix@C=.5pc@R=0.5pc{
   1:ins \ar[d] \ar[r]& 3:cons \ar[d] \ar[r] & 4:\bullet \\ 
   2:\bullet &          5: \bullet & 6: cons \ar[ul] \ar[r] & 7 : \bullet}$  
 & $\xymatrix@C=.5pc@R=0.5pc{ 11:cons \ar[r] \ar[d] & 3:cons \ar[d] \ar[r] & 4:\bullet \\
               2:\bullet & 5:\bullet & 6: cons \ar[ull] \ar[r] & 
               7 : \bullet
   }$ 
\\  \hline
   $\xymatrix@C=.5pc@R=0.5pc{
     2:e & 1:ins \ar[l] \ar[dl] & 0:h \ar[l] \ar[dll] \\
     3:cons  \ar[r] \ar[d] & 4:cons  \ar[d] \ar[r] & 8:cons  \ar[d]\ar[r]
                         & 6:cons \ar[d] \ar@/^-1pc/[lll]\\
     5:a & 9:b & 10:c & 7:d 
   }$
 & $\xymatrix@C=.5pc@R=0.5pc{2:e & 11:cons \ar[l] \ar[dl] & 
     0:h \ar@/^1pc/[l] \ar@/^-.5pc/[l] \\
     3:cons  \ar[r] \ar[d] & 4:cons  \ar[d] \ar[r] & 8:cons  \ar[d]\ar[r]
                         & 6:cons \ar[d] \ar[llu]\\
     5:a & 9:b & 10:c & 7:d        
   }$ \\ \hline  
\end{tabular}
\end{center}
%\end{example}

%\begin{example}[Append]
\section*{Appending linked lists}
  \label{ex:append}

\begin{minipage}{0.66\linewidth}

  We now consider the rules for the operation ``$+$'' which appends
  two linked lists. The lists are supposed to be built with the
  constructors $cons$, and $nil$.
The base case is defined when the first argument is $nil$ as in the rule opposite. 
\end{minipage}
\hspace{0.2cm}
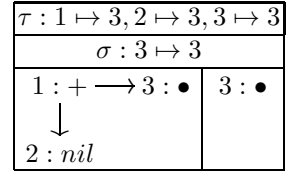
\begin{minipage}{0.29\linewidth}
%\begin{center}
 \begin{tabular}{|c|c|}
 \hline
 \multicolumn{2}{|c|}{$\tau :1 \mapsto 3, 2 \mapsto 3, 3 \mapsto 3$} \\ \hline 
 \multicolumn{2}{|c|}{$\sigma : 3 \mapsto 3$} \\ \hline
   $\xymatrix@C=1pc@R=1pc{1:+ \ar[d] \ar[r] &  3:\bullet \\
              2:nil}$ 
   & $\xymatrix{3:\bullet}$ \\ \hline
 \end{tabular}
%\end{center}
\end{minipage}

\begin{minipage}{0.50\linewidth}
When the first argument of $+$ is a list different from $nil$, we call 
an auxiliary function denoted ``$+1$'', of arity~3. The role of this function 
is to go through the first list until the end and concatenate the two lists just by pointer redirection.
The first call to the operation $+1$ is done by the rule opposite: 
\end{minipage}
\hspace{0.2cm}
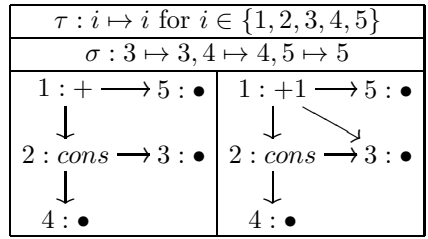
\begin{minipage}{0.40\linewidth}
%\begin{center}
 \begin{tabular}{|c|c|}
    \hline
    \multicolumn{2}{|c|}{$\tau : i \mapsto i \mbox{ for } i  \in \{1,2,3,4,5 \}$} \\ \hline 
    \multicolumn{2}{|c|}{$\sigma : 3 \mapsto 3, 
                          4 \mapsto 4, 5 \mapsto 5$} \\ \hline
   $\xymatrix@C=1pc@R=1pc{1 : + \ar[d] \ar[r] & 5:\bullet \\
              2 : cons \ar[d] \ar[r] & 3:\bullet \\
              4 : \bullet}$ 
   & $\xymatrix@C=1pc@R=1pc{1 : +1 \ar[d] \ar[dr] \ar[r]
                        & 5 : \bullet \\
                2 : cons \ar[d] \ar[r] & 3 : \bullet \\
                4 : \bullet}$ \\ \hline
 \end{tabular}
%\end{center}
\end{minipage}

The second argument of $+1$ is used to go through the list starting 
at node $2$ to get the last element of the list. This is implemented by the following rule~:
%The following rule is used to make a step beyond in the list. It is 
%simply implemented by moving the second argument of $+1$ from node 
%$3$ to node $4$. 
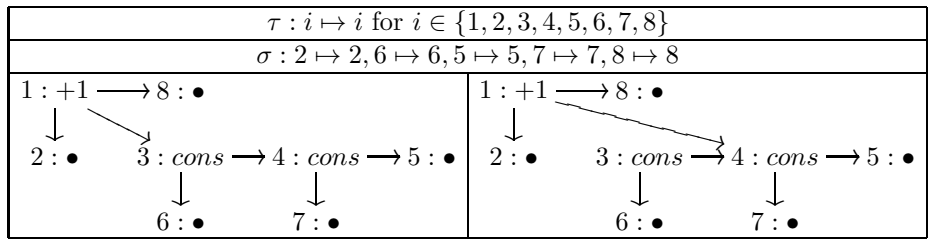
\begin{center}
 \begin{tabular}{|c|c|}
    \hline
    \multicolumn{2}{|c|}{$\tau : i \mapsto i \mbox{ for } i  \in \{1,2,3,4,5,6,7,8 \}$} \\ \hline 
    \multicolumn{2}{|c|}{$\sigma : 2 \mapsto 2, 
                          6 \mapsto 6, 5 \mapsto 5, 
                          7 \mapsto 7, 8 \mapsto 8$} \\ \hline
     $\xymatrix@C=1pc@R=1pc{1: +1 \ar[d] \ar[dr] \ar[r] & 8:\bullet \\ 
                2: \bullet & 3 : cons \ar[r] \ar[d] 
                & 4:cons \ar[r]\ar[d] & 5:\bullet \\
                & 6:\bullet & 7:\bullet}$
    &$\xymatrix@C=1pc@R=1pc{1: +1 \ar[d] \ar[drr] \ar[r] & 8:\bullet \\
                         2:\bullet & 3 : cons \ar[d] \ar[r]
                        & 4:cons \ar[d] \ar[r] & 5:\bullet \\
                        & 6:\bullet & 7:\bullet}$ \\ \hline
 \end{tabular}
\end{center}

The last case for operation $+1$, is implemented as follows. We simply 
redirect the second edge from $3$ to $4$ (which is $nil$) towards
$6$ (e.g., $\tau(4) = 6$), which is the head of the second list to append. 
The overall result of the operation $+1$, is the head of first list, node $2:\bullet$. This is implemented by $\tau(1) = 2$.
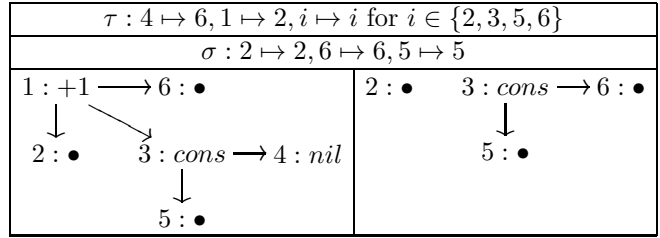
\begin{center}
 \begin{tabular}{|c|c|}
    \hline
    \multicolumn{2}{|c|}{$\tau : 4 \mapsto 6, 1 \mapsto 2, i \mapsto i \mbox{ for } i  \in \{2,3,5,6 \}$} \\ \hline 
    \multicolumn{2}{|c|}{$\sigma : 2 \mapsto 2, 
                          6 \mapsto 6, 5 \mapsto 5$} \\ \hline
     $\xymatrix@C=1pc@R=1pc{1:+1\ar[d] \ar[dr] \ar[r] & 6:\bullet \\ 
                2: \bullet & 3 : cons \ar[r] \ar[d] & 4:nil \\
                & 5:\bullet}$
    &$\xymatrix@C=1pc@R=1pc{2:\bullet & 3:cons \ar[r] \ar[d] 
                            & 6:\bullet \\
                                & 5:\bullet}$ \\ \hline 
 \end{tabular}
\end{center}
%\end{example}

%\begin{example}[Memory freeing]
\section*{Memory freeing}
  \label{ex:memfree} 
In this example we show how we can free the
  memory used by a circular list. As we are concerned with termgraphs
  where every function symbol has a fixed arity, it is not possible to
  create dangling pointers nor to remove useless pointers. This
  constraint is expressed by the fact that every node in a left-hand
  side $L$ must have an image in the right-hand side $R$ by $\tau$.
  
  The operation $free$ has two arguments. The first one is a
  particular node labeled by a constant \emph{null}. This constant is
  dedicated to be the target of the edges which were pointing the
  freed nodes. The second argument of $free$ is the list of cells to
  be freed.

Below, we give a rule defining the operation $free$ in the case of a
list with at least two different elements. We also illustrate its
application on a list of length two.

%We give the rule for the general case (list of size $>$ 1) below 
%together with an example of a rewrite step for a list of size 2.

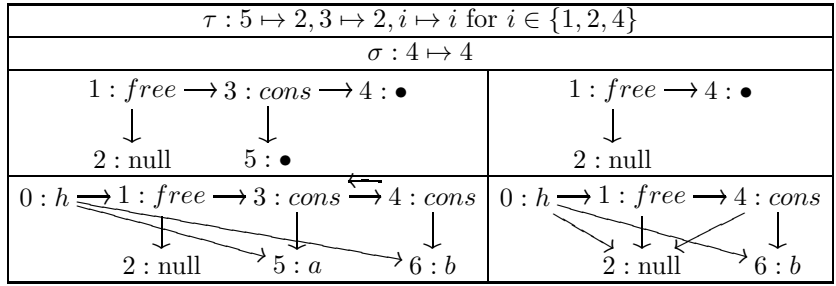
\begin{center}
\begin{tabular}{|c|c|}
    \hline
    \multicolumn{2}{|c|}{$\tau : 5 \mapsto 2, 3 \mapsto 2,i \mapsto i \mbox{ for } i  \in \{1,2,4\} $} \\ \hline 
    \multicolumn{2}{|c|}{$\sigma : 4 \mapsto 4$} \\ \hline
 $\xymatrix@C=1pc@R=1pc{1:free \ar[d] \ar[r] & 
             3:cons \ar[d] \ar[r] & 4 : \bullet \\
             2:\mbox{null} & 5:\bullet 
   }$ 
 & $\xymatrix@C=1pc@R=1pc{ 1:free \ar[d] \ar[r] & 4:\bullet\\
               2:\mbox{null}
   }$ 
\\ \hline 
   $\xymatrix@C=1pc@R=1pc{
    0:h \ar[r] \ar[drr] \ar[drrr] & 1:free \ar[d] \ar[r] 
             & 3:cons \ar[d] \ar[r] 
             & 4 : cons \ar[d] \ar@/^-.5pc/[l] \\
             & 2:\mbox{null} & 5:a & 6:b 
   }$
 & $\xymatrix@C=1pc@R=1pc{0:h\ar[r] \ar[dr] \ar[drr] & 
                                       1:free \ar[d] \ar[r] 
                                     & 4:cons \ar[d] \ar[dl] \\
               & 2:\mbox{null} & 6:b       
   }$ \\ \hline
\end{tabular}
\end{center}
Notice that pointers incident to nodes $3$ and $5$
are redirected towards $2$. % which acts like a lightning rod.

There are two cases for lists with one element. The following rule
specifies the case where the last element of the list is obtained
after freeing other elements of the list.  We illustrate the rewrite
rule on the graph obtained earlier (up to renaming of nodes).

%Notice that because of injectivity condition on the matching, these
%two rules cannot match the same graph.

%In this rule only the node labeled by
%$null$ is obtained as a result in the right-hand side.  

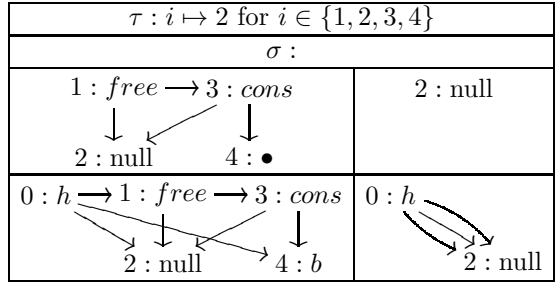
\begin{center}
 \begin{tabular}{|c|c|}
    \hline
    \multicolumn{2}{|c|}{$\tau : i \mapsto 2 \mbox{ for } i  \in \{1,2,3,4\}$} \\ \hline 
    \multicolumn{2}{|c|}{$\sigma : $} \\ \hline
  $\xymatrix@C=1pc@R=1pc{ 1:free \ar[d] \ar[r] & 3:cons \ar[dl] \ar[d]  \\
             2:\mbox{null} & 4:\bullet
     }$ 
   & $\xymatrix@C=1pc@R=1pc{ 2:\mbox{null} }$  \\ \hline
  $\xymatrix@C=1pc@R=1pc{0:h\ar[r] \ar[dr] \ar[drr] 
                                     & 1:free \ar[d] \ar[r] 
                                     & 3 :cons \ar[d] \ar[dl] \\
               & 2:\mbox{null} & 4:b       
   }$ &
   $\xymatrix@C=1pc@R=1pc{0:h \ar@/^-.5pc/[dr] \ar@/^.5pc/[dr] \ar[dr] \\
               & 2:\mbox{null}        
   }$ \\ \hline
 \end{tabular}
\end{center}

Finally, because of the injectivity condition on matching, we have 
to consider the special case of lists of size one. This is done by 
the following rule:

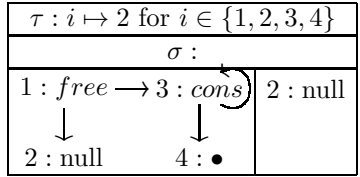
\begin{center}
 \begin{tabular}{|c|c|}
    \hline
    \multicolumn{2}{|c|}{$\tau : i \mapsto 2 \mbox{ for } i  \in \{1,2,3,4\}$} \\ \hline 
    \multicolumn{2}{|c|}{$\sigma : $} \\ \hline
    $\xymatrix@C=1pc@R=1pc{1:free \ar[r] \ar[d]
                           & 3:cons \ar[d] \ar@(dr,ur)[]\\
                           2:\mbox{null} & 4:\bullet}$ & 
    $\xymatrix@C=1pc@R=1pc{2:\mbox{null}}$\\ \hline

 \end{tabular}
\end{center}

%\end{example}

% conclusion...
\section{Conclusion}
\label{sec:conc}

% what we have done

We have proposed a new way to define termgraph rewrite rules. Rules
are quite simple. A rule is a tuple $(L,R,\tau, \sigma)$ where $L$ and
$R$ are termgraphs representing the left-hand and the right-hand sides
of the rule, $\tau$ is a mapping from the nodes of $L$ to those of $R$
and $\sigma$ is a partial function from nodes of $R$ to nodes of $L$. $\tau$
describes how incident edges of the nodes in $L$ are connected in
$R$. It should be noted that $\tau$ is not required to be a graph
morphism as in the classical algebraic graph transformation approaches
\cite{CorradiniMREHL97,EhrigHKLRWC97}. As for $\sigma$, it is useful
only when one needs to clone some parts of $L$. We defined rewrite steps as
 pushouts in an appropriate category as shown in section~\ref{sec:clo}.

The proposed rewrite systems offer the possibility
to transform  cyclic termgraphs either by performing local edge redirections
or global edge redirections, as defined  in
\cite{DuvalEP07} following a DPO approach, but provides also new features not present in
\cite{DuvalEP07} such as cloning or deletion of nodes.

%Cloning is also one of the features of the sesqui-pushout approach to
%graph transformation \cite{CorradiniHHK06}. This approach coincides
%with the DPO on left linear rules ... 

%comparaison DPO, SPO, sesqui, duvaletal, cloning et oubli de noeuds,
%graph programming plump, survey plump, ariola

Besides the algorithmic approaches to termgraph transformation
(e.g. \cite{BVG87}), a categorical framework dedicated to
\emph{cyclic} termgraph transformation could be found in
\cite{CorradiniG97} where the authors propose, following
\cite{Power89}, a 2-categorical presentation of termgraph
rewriting. They almost succeeded to represent the full operational
view of termgraph rewriting as defined in \cite{BVG87}, but differ on
rewriting circular redexes. For example, the application of the
rewrite rule $f(x) \to x$ on the termgraph $n:f(n)$ yields the same
termgraph (i.e. $n:f(n)$) according to \cite{BVG87} but yields an
unlabeled node, say $p:\bullet$, according to \cite{CorradiniG97}.  The
definition of rewrite rules that we propose in this paper allows us to
make a clear distinction between the two behaviours.  The rule
$(n:f(m:\bullet), p:\bullet, \tau =\{n \mapsto p, m \mapsto p \},
\sigma = \{\})$ behaves as in \cite{CorradiniG97} when applied on
$n:f(n)$, whereas the behaviour described in
\cite{BVG87} can be obtained by simply declaring that node $p$ is a
clone of node $m$ via $\sigma$ as in the following rule
$(n:f(m:\bullet), p:\bullet, \tau =\{n \mapsto p, m \mapsto p \},
\sigma = \{p \mapsto m \})$.

Future works include the generalization of the proposed systems to
other graphs less constrained than termgraphs.  This would allow us
to require from $\tau$, in a rule $(L,R,\tau,\sigma)$, to be a partial
function like in the single pushout approach \cite{EhrigHKLRWC97}.

%futur work
%\bibliographystyle{abbrv}
%\bibliography{biblio}

\begin{thebibliography}{10}

\bibitem{ArK96}
Z.~Ariola and J.~Klop.
\newblock Equational term graph rewriting.
\newblock {\em Fundamenta Informaticae}, 26(3-4), 1996.

\bibitem{BVG87}
H.~Barendregt, M.~van Eekelen, J.~Glauert, R.~Kenneway, M.~J. Plasmeijer, and
  M.~Sleep.
\newblock Term graph rewriting.
\newblock In {\em PARLE'87}, pages 141--158. Springer Verlag LNCS 259, 1987.

\bibitem{CorradiniG97}
A.~Corradini and F.~Gadducci.
\newblock A 2-categorical presentation of term graph rewriting.
\newblock In {\em 7th International Conference on Category Theory and Computer
  Science (CTCS 97)}, volume 1290 of {\em Lecture Notes in Computer Science},
  pages 87--105. Springer, 1997.

\bibitem{CorradiniMREHL97}
A.~Corradini, U.~Montanari, F.~Rossi, H.~Ehrig, R.~Heckel, and M.~L{\"o}we.
\newblock Algebraic approaches to graph transformation - part {I}: Basic
  concepts and double pushout approach.
\newblock In {\em Handbook of Graph Grammars}, pages 163--246, 1997.

\bibitem{DuvalEP07}
D.~Duval, R.~Echahed, and F.~Prost.
\newblock Modeling pointer redirection as cyclic term-graph rewriting.
\newblock {\em Electr. Notes Theor. Comput. Sci.}, 176(1):65--84, 2007.

\bibitem{EhrigHKLRWC97}
H.~Ehrig, R.~Heckel, M.~Korff, M.~L{\"o}we, L.~Ribeiro, A.~Wagner, and
  A.~Corradini.
\newblock Algebraic approaches to graph transformation - part ii: Single
  pushout approach and comparison with double pushout approach.
\newblock In {\em Handbook of Graph Grammars}, pages 247--312, 1997.

\bibitem{EhrigPS73}
H.~Ehrig, M.~Pfender, and H.~J. Schneider.
\newblock Graph-grammars: An algebraic approach.
\newblock In {\em 14th Annual Symposium on Foundations of Computer Science
  (FOCS), 15-17 October 1973, The University of Iowa, USA}, pages 167--180.
  IEEE, 1973.

\bibitem{Kennaway87}
R.~Kennaway.
\newblock On ``on graph rewritings''.
\newblock {\em Theor. Comput. Sci.}, 52:37--58, 1987.

\bibitem{Lowe93}
M.~L{\"o}we.
\newblock Algebraic approach to single-pushout graph transformation.
\newblock {\em Theor. Comput. Sci.}, 109(1{\&}2):181--224, 1993.

\bibitem{maclane}
S.~Mac~Lane.
\newblock {\em Categories for the Working Mathematician}, volume~5.
\newblock Springer-Verlag, second edition edition, 1998.

\bibitem{Plu98a}
D.~Plump.
\newblock Term graph rewriting.
\newblock In H.~Ehrig, G.~Engels, H.~J. Kreowski, and G.~Rozenberg, editors,
  {\em Handbook of Graph Grammars and Computing by Graph Transformation},
  volume~2, pages 3--61. World Scientific, 1999.

\bibitem{Power89}
A.~J. Power.
\newblock An abstract formulation for rewrite systems.
\newblock In {\em Category Theory and Computer Science}, volume 389 of {\em
  Lecture Notes in Computer Science}, pages 300--312. Springer, 1989.

\bibitem{Rao84}
J.~C. Raoult.
\newblock On graph rewriting.
\newblock {\em Theoretical Computer Science}, 32:1--24, 1984.

\bibitem{SPV93}
M.~R. Sleep, M.~J. Plasmeijer, and M.~C. J.~D. van Eekelen, editors.
\newblock {\em Term Graph Rewriting. Theory and Practice}.
\newblock J. Wiley \& Sons, Chichester, UK, 1993.

\end{thebibliography}

\end{document}